\documentclass[twocolumn,nofootinbib,longbibliography]{revtex4-1}

\usepackage[english]{babel}
\usepackage[toc]{appendix}
\usepackage{natbib}
\usepackage{hyperref}
\usepackage[multiple]{footmisc}
\usepackage[threshold=1, autopunct, autostyle]{csquotes}

\usepackage{amsmath}
\usepackage{mathrsfs}
\usepackage{amssymb}
\usepackage{amsfonts}
\usepackage{dsfont}
\usepackage{IEEEtrantools}
\usepackage{bm}
\usepackage{amsthm}

\usepackage{xstring}
\usepackage{catchfile}

\makeatletter
\@ifclassloaded{beamer}
  {}
  {
\usepackage[inline]{enumitem}
    \bibpunct{[}{]}{,}{n}{}{,}

  }
\makeatother

\renewcommand{\H}{\ensuremath{\mathcal{H}}}

\DeclareMathOperator{\Prb}{P}

\newcommand{\cProb}[2]{\Prb(#1\!\mid\!#2)}

\newcommand{\Bt}{{\tt t}}
\newcommand{\Bf}{{\tt f}}

\newcommand{\Proj}[1]{\ensuremath{\mathfrak{P}(#1)}}
\newcommand{\ProjH}{\Proj{\H}}

\newcommand{\Q}{\ensuremath{\mathfrak{Q}}}

\newcommand{\QC}{\ensuremath{\mathcal{Q}}}
\newcommand{\Z}{\ensuremath{\mathcal{Z}}}

\usepackage{pgfplots}
\usepackage{tikz}
\usetikzlibrary{shapes,calc,positioning,arrows.meta,fadings}

\pgfdeclarelayer{bg}    \pgfsetlayers{bg,main}

\tikzset{colorbox/.style={thick, rounded corners=2pt, text height=1.7ex,text depth=.25ex, draw=#1!70!black, fill=#1!30}}
\tikzset{colorboxS/.style={thick, rounded corners=2pt, text height=1.2ex,text depth=.25ex, draw=#1!70!black, fill=#1!30, font=\footnotesize}}
\tikzset{colorboxXS/.style={thick, rounded corners=2pt, text height=0.7ex,text depth=.10ex, draw=#1!70!black, fill=#1!30, font=\tiny}}
\tikzset{roundedbox/.style={thick, rounded corners=2pt, text height=1.7ex,text depth=.25ex, draw=black}}

\tikzset{CBedgy/.style={thick, text height=1.7ex,text depth=.25ex, draw=#1!70!black, fill=#1!30, font=\ttfamily}}
\tikzset{colorelement/.style={thick, rounded corners=2pt, draw=#1!70!black, fill=#1!30}}

\tikzset{tb/.style={font=\footnotesize, text width=#1, text opacity=1, opacity=.8, fill=white},
  tb/.default=4.5cm}
\tikzset{tbS/.style={font=\scriptsize, text width=#1, text opacity=1, opacity=.8, fill=white},
  tbS/.default=4.5cm}
\tikzset{tbXS/.style={font=\tiny, text width=#1, text opacity=1, opacity=.8, fill=white},
  tbXS/.default=4.5cm}
\tikzset{tbBB/.style={draw=black!70, rounded corners=2pt, font=\footnotesize, text width=#1, text opacity=1, opacity=.8, fill=white},
  tbBB/.default=4.5cm}
\tikzset{tbBBS/.style={draw=black!70, rounded corners=2pt, font=\scriptsize, text width=#1, text opacity=1, opacity=.8, fill=white},
  tbBBS/.default=4.5cm}
\tikzset{tbBBXS/.style={draw=black!70, rounded corners=2pt, font=\tiny, text width=#1, text opacity=1, opacity=.8, fill=white},
  tbBBXS/.default=4.5cm}

\tikzset{quote/.style={thick, draw=black!70, rounded corners=2pt, font=\footnotesize, text width=#1, text opacity=1, opacity=.8, fill=white},
  quote/.default=4.5cm}
\tikzset{quoteS/.style={thick, draw=black!70, rounded corners=2pt, font=\scriptsize, text width=#1, text opacity=1, opacity=.8, fill=white},
  quoteS/.default=4.5cm}
\tikzset{quoteXS/.style={thick, draw=black!70, rounded corners=2pt, font=\tiny, text width=#1, text opacity=1, opacity=.8, fill=white},
  quoteXS/.default=4.5cm}
\tikzset{quoteNB/.style={font=\footnotesize, text width=#1, text opacity=1, opacity=.8, fill=white},
  quoteNB/.default=4.5cm}

\tikzset{comment/.style={thick, draw=black!70, rounded corners=2pt, font=\scriptsize\itshape, text width=#1, text opacity=1, opacity=.8, fill=white},
  comment/.default=4.5cm}
\tikzset{commentS/.style={thick, draw=black!70, rounded corners=2pt, font=\tiny\itshape, text width=#1, text opacity=1, opacity=.8, fill=white},
  commentS/.default=4.5cm}
\tikzset{commentVW/.style={thick, draw=black!70, rounded corners=2pt, font=\scriptsize\itshape, text opacity=1, opacity=.8, fill=white}}
\tikzset{commentSVW/.style={thick, text height=0.7ex,text depth=.10ex, draw=black!70, rounded corners=2pt, font=\tiny\itshape, text opacity=1, opacity=.8, fill=white}}
\tikzset{commentSVWR/.style={thick, text height=0.7ex,text depth=.10ex, draw=black!70, rounded corners=2pt, font=\tiny, text opacity=1, opacity=.8, fill=white}}

\tikzset{conn/.style={thick, shorten <=#1, shorten >=#1}}
\tikzset{tconn/.style={shorten <=#1, shorten >=#1}}
\tikzset{gconn/.style={thick, shorten <=#1, shorten >=#1, draw=gray!80}}

\tikzset{arr_node/.style={pos=0.5,above,font=\scriptsize, sloped}}
\tikzset{ctag/.style={thick, dashed, rounded corners=2pt, text height=1.7ex,text depth=.25ex, draw=#1!70!black, fill=#1!30, font=\ttfamily}}

\tikzset{ist/.style={thick, shorten <=4pt, shorten >=4pt, arrows = {-Bracket[reversed,round]}}}
\tikzset{istgleich/.style={thick, shorten <=4pt, shorten >=4pt, arrows = {Bracket[reversed,round]-Bracket[reversed,round]}}}
\tikzset{nicht/.style={thick, dashed, shorten <=4pt, shorten >=4pt, arrows = {Bracket[round]-}}}

\definecolor{yorange}{HTML}{ff8c00}

\newcommand\blArrow{}
\def\blArrow[#1](#2);
  { \draw[#1] (#2) -- ++(.9,0) -- ++(-.9,-5) -- cycle; }
\newcommand\stdBlArrow{}
\def\stdBlArrow(#1)
  { \blArrow[line width=2pt, draw=red!70!black, fill=red!40, rounded corners=1pt](#1) }

\newcommand\redArrow{}
\def\redArrow(#1);
  { \draw[thick] ($(#1)+(-.3,.1)$) -- ($(#1)+(0,-.1)$) -- ($(#1)+(.3,.1)$);}
\newcommand\fSTS{}
\def\fSTS(#1);
  { \draw[fill=red!40, thick, draw=red!60!black] (#1) circle (.5);
    \draw[->] ($(#1)+(-.05,-.05)$) to ++(0,.4);
    \draw[->] ($(#1)+(-.05,-.05)$) to ++(.4,0);
    \draw[->] ($(#1)+(-.05,-.05)$) to ++(-.2,-.2);}
\newcommand\reduction{}
\def\reduction(#1);
  { \draw[fill=cyan!40, thick, draw=cyan!60!black] (#1) circle (.5);
    \draw[fill=black] ($(#1) + (-.2, .25)$) to ++(0.4,0) to ++(-.3,-.6) to cycle;}
\newcommand\customLegend{}
\def\customLegend(#1);
  {\fSTS(#1);
   \node[right of=#1] {fixed space-time structure};
   \draw[rounded corners=2pt] ($(#1)+(-.5,-.5)$) rectangle ($(#1)+(4, -4)$);}

\makeatletter

\newcommand*{\blitzset}{\pgfqkeys{/blitz}}\blitzset{
  height/.initial=4cm,
  width/.initial=2cm,
  breadth/.initial=.3cm,
  ratio/.initial=.4,
}

\pgfdeclareshape{blitz}
{
  \savedanchor\centerpoint{
    \pgf@x = .5\wd\pgfnodeparttextbox
    \pgf@y = .5\ht\pgfnodeparttextbox
  }
  \anchor{center}{\centerpoint}
  \anchor{north}{\centerpoint\pgfmathsetlength\pgf@ya{\pgfkeysvalueof{/blitz/ratio}*\pgfkeysvalueof{/blitz/height}}\advance \pgf@y by \pgf@ya}
\anchor{south}{\centerpoint\pgfmathsetlength\pgf@ya{-(1-\pgfkeysvalueof{/blitz/ratio})*\pgfkeysvalueof{/blitz/height}}\advance \pgf@y by \pgf@ya}
  \anchor{east}{\centerpoint\pgfmathsetlength\pgf@ya{\pgfkeysvalueof{/blitz/breadth}}\advance \pgf@y by \pgf@ya\pgfmathsetmacro{\alpha}{atan(2*\pgfkeysvalueof{/blitz/ratio}*\pgfkeysvalueof{/blitz/height}/\pgfkeysvalueof{/blitz/width})}\pgfmathsetlength\pgf@xa{.5*\pgfkeysvalueof{/blitz/width}+\pgfkeysvalueof{/blitz/breadth}/tan(\alpha/2)}\advance\pgf@x by \pgf@xa}
  \anchor{west}{\centerpoint\pgfmathsetlength\pgf@ya{\pgfkeysvalueof{/blitz/breadth}}\advance \pgf@y by -\pgf@ya\pgfmathsetmacro{\alpha}{atan(2*\pgfkeysvalueof{/blitz/ratio}*\pgfkeysvalueof{/blitz/height}/\pgfkeysvalueof{/blitz/width})}\pgfmathsetlength\pgf@xa{.5*\pgfkeysvalueof{/blitz/width}+\pgfkeysvalueof{/blitz/breadth}/tan(\alpha/2)}\advance\pgf@x by -\pgf@xa}

  \backgroundpath{
    \centerpoint \pgf@xa=\pgf@x \pgf@ya=\pgf@y
    \pgfmathsetmacro{\alpha}{atan(2*\pgfkeysvalueof{/blitz/ratio}*\pgfkeysvalueof{/blitz/height}/\pgfkeysvalueof{/blitz/width})}
    \pgfmathsetlength\pgf@xb{.5*\pgfkeysvalueof{/blitz/width}-\pgfkeysvalueof{/blitz/breadth}/tan(\alpha/2)}
    \pgfmathsetlength\pgf@yb{\pgfkeysvalueof{/blitz/breadth}}
    \advance \pgf@xa by \pgf@xb
    \advance \pgf@ya by -\pgf@yb
\pgfpathmoveto{\pgfpoint{\pgf@xa}{\pgf@ya}}
    \pgfmathsetlength\pgf@xb{\pgfkeysvalueof{/blitz/width}}
    \advance \pgf@xa by -\pgf@xb
    \pgfpathlineto{\pgfpoint{\pgf@xa}{\pgf@ya}}
    \pgfmathsetlength\pgf@yb{\pgfkeysvalueof{/blitz/breadth}+\pgfkeysvalueof{/blitz/ratio}*\pgfkeysvalueof{/blitz/height}}
    \pgfmathsetlength\pgf@xb{\pgf@yb*sin(90-\alpha)}
    \advance \pgf@xa by \pgf@xb
    \advance \pgf@ya by \pgf@yb
    \pgfpathlineto{\pgfpoint{\pgf@xa}{\pgf@ya}}
    \pgfmathsetlength\pgf@xb{2*\pgfkeysvalueof{/blitz/breadth}/cos(90-\alpha)}
    \advance \pgf@xa by \pgf@xb
    \pgfpathlineto{\pgfpoint{\pgf@xa}{\pgf@ya}}
    \pgfmathsetlength\pgf@yb{\pgfkeysvalueof{/blitz/ratio}*\pgfkeysvalueof{/blitz/height}-\pgfkeysvalueof{/blitz/breadth}}
    \pgfmathsetlength\pgf@xb{\pgf@yb*sin(90-\alpha)}
    \advance \pgf@xa by -\pgf@xb
    \advance \pgf@ya by -\pgf@yb
    \pgfpathlineto{\pgfpoint{\pgf@xa}{\pgf@ya}}
    \pgfmathsetlength\pgf@xb{\pgfkeysvalueof{/blitz/width}}
    \advance \pgf@xa by \pgf@xb
    \pgfpathlineto{\pgfpoint{\pgf@xa}{\pgf@ya}}
    \pgfmathsetlength\pgf@yb{(1-\pgfkeysvalueof{/blitz/ratio})*\pgfkeysvalueof{/blitz/height}+\pgfkeysvalueof{/blitz/breadth}}
    \pgfmathsetlength\pgf@xb{.5*\pgfkeysvalueof{/blitz/width}+\pgfkeysvalueof{/blitz/breadth}/tan(\alpha/2)}
    \advance \pgf@xa by -\pgf@xb
    \advance \pgf@ya by -\pgf@yb
    \pgfpathlineto{\pgfpoint{\pgf@xa}{\pgf@ya}}
    \pgfpathclose
  }
}

\newcommand*{\srefset}{\pgfqkeys{/sref}}\srefset{
  height/.initial=.5cm,
  width/.initial=.5cm,
}

\pgfdeclareshape{sref}
{
  \savedanchor\centerpoint{
    \pgf@x = .5\wd\pgfnodeparttextbox
    \pgf@y = .5\ht\pgfnodeparttextbox
  }
  \anchor{center}{\centerpoint}
  \anchor{north}{\centerpoint\pgfmathsetlength\pgf@ya{.5*\pgfkeysvalueof{/sref/height}}\advance \pgf@y by \pgf@ya}
  \anchor{south}{\centerpoint\pgfmathsetlength\pgf@ya{-.5*\pgfkeysvalueof{/sref/height}}\advance \pgf@y by \pgf@ya}
  \anchor{east}{\centerpoint\pgfmathsetlength\pgf@xa{.5*\pgfkeysvalueof{/sref/width}}\advance \pgf@x by \pgf@xa}
  \anchor{west}{\centerpoint\pgfmathsetlength\pgf@xa{-.5*\pgfkeysvalueof{/sref/width}}\advance \pgf@x by \pgf@xa}

  \backgroundpath{
    \centerpoint \pgf@xa=\pgf@x \pgf@ya=\pgf@y
    \pgfmathsetlength\pgf@xb{.5*\pgfkeysvalueof{/sref/width}}
    \pgfmathsetlength\pgf@yb{.5*\pgfkeysvalueof{/sref/height}}
    \pgfmathsetlength\pgf@yc{.75*\pgfkeysvalueof{/sref/height}}
    \advance \pgf@xa by \pgf@xb
    \pgfpathmoveto{\pgfpointpolar{20}{\pgf@xa}}
    \pgfpatharc{20}{170}{\pgf@xb}
    \pgfpathlineto{\pgfpointpolar{140}{\pgf@yc}}
    \pgfpathmoveto{\pgfpointpolar{200}{\pgf@xa}}
    \pgfpatharc{200}{350}{\pgf@xb}
    \pgfpathlineto{\pgfpointpolar{320}{\pgf@yc}}
  }
}

\newcommand*{\qmarkset}{\pgfqkeys{/qmark}}\qmarkset{
  height/.initial=10pt,
  width/.initial=10pt,
  ratio/.initial=.5,
}

\pgfdeclareshape{qmark}{
  \savedanchor\centerpoint{\pgf@x=.5\wd\pgfnodeparttextbox \pgf@y=.5\ht\pgfnodeparttextbox }
  \anchor{center}{\centerpoint}
  \anchor{north}{\centerpoint\pgfmathsetlength\pgf@ya{\pgfkeysvalueof{/blitz/ratio}*\pgfkeysvalueof{/blitz/height}}\advance \pgf@y by \pgf@ya}
  \anchor{south}{\centerpoint\pgfmathsetlength\pgf@ya{-(1-\pgfkeysvalueof{/blitz/ratio})*\pgfkeysvalueof{/blitz/height}}\advance \pgf@y by \pgf@ya}
  \anchor{west}{\centerpoint\pgfmathsetlength\pgf@xa{-.5*\pgfkeysvalueof{/blitz/height}}\advance \pgf@x by \pgf@xa}
  \anchor{east}{\centerpoint\pgfmathsetlength\pgf@xa{.5*\pgfkeysvalueof{/blitz/height}}\advance \pgf@x by \pgf@xa}

  \backgroundpath{
    \pgfkeysgetvalue{/qmark/width}{\width}
    \pgfkeysgetvalue{/qmark/height}{\height}
    \pgfkeysgetvalue{/qmark/ratio}{\r}
    \pgfpathmoveto{\pgfpointadd{\centerpoint}{\pgfpoint{-.5*\width}{-.1*\height}}}
    \pgfpathcurveto
{\pgfpointadd{\centerpoint}{\pgfpoint{-3/6*\width}{.3*\height}}}
      {\pgfpointadd{\centerpoint}{\pgfpoint{.5*\width}{.4*\height}}}
      {\pgfpointadd{\centerpoint}{\pgfpoint{.5*\width}{-.05*\height}}}
    \pgfpathcurveto
{\pgfpointadd{\centerpoint}{\pgfpoint{.5*\width}{-.45*\height}}}
      {\pgfpointadd{\centerpoint}{\pgfpoint{-.17*\width}{-.2*\height}}}
      {\pgfpointadd{\centerpoint}{\pgfpoint{-.2*\width}{-0.78*\height}}}
    
    \pgfpathcircle{\pgfpointadd{\centerpoint}{\pgfpoint{-.2*\width}{-1.0*\height}}}{.08*\width}
  }
}

\makeatother

\makeatletter
\@ifclassloaded{beamer}
  {}
  {
    \newtheorem{theorem}{Theorem}

    \newtheorem{lemma}{Lemma}
    
\theoremstyle{remark}
    
    \theoremstyle{definition}
    \newtheorem{definition}{Definition}

  }
\makeatother

\newif\ifradical
\radicaltrue
\newif\ifbeta
\betatrue
\newif\ifuncertain
\uncertaintrue
\newif\ifcomments
\commentstrue

\newcounter{CtrSprachspiel}

\renewcommand{\epsilon}{\varepsilon}

\makeatletter\begin{document}

\title{Contextuality: It's a Feature, not a Bug}
\title{Measuring measurements}
\title{Measuring Measuring}

\author{Arne Hansen and Stefan Wolf}
\affiliation{Facolt\`a di Informatica, 
Universit\`a della Svizzera italiana, Via G. Buffi 13, 6900 Lugano, Switzerland}

\date{\today}

\begin{abstract}
\noindent
Measurements play a crucial role in \emph{doing physics}:
Their results provide the basis on which we adopt or reject physical theories.
In this note, we examine the effect of subjecting measurements themselves to our experience.
We require that our contact with the world is empirically warranted.
Therefore, we study theories that satisfy the following assumption: Interactions are accounted for so that they are \emph{empirically traceable}, and \emph{observations necessarily go with such an interaction} with the observed system.
Examining, with regard to these assumptions, an abstract representation of measurements with tools from quantum logic leads us to \emph{contextual} theories.
Contextuality becomes a means to render interactions, thus also measurements, empirically tangible. 
The measurement becomes problematic---also beyond quantum mechanics---if one tries to commensurate the assumption of tangible interactions with the notion of a \emph{spectator theory}, i.e., with the idea that measurement results are \emph{read off without effect}.
The problem, thus, presents itself as the collision of different epistemological stances with repercussions beyond quantum mechanics.
 \end{abstract}

\maketitle

\section{Introduction}
\label{sec:introduction}
\noindent
The infamous Wigner's-friend experiment~\cite{wigner1963problem, deutsch1985quantum, Wigner1961} serves to illustrate the measurement problem~\cite{Maudlin95,BHW16}:
An experimenter, called Wigner, performs a measurement on the joint system~$F\otimes S$ that consists of his friend~$F$ who measures a system~$S$.
If the joint system~$F\otimes S$ is isolated,~$S$ is prepared in a suitable state, and Wigner and his friend choose specific measurements, then there are different predictions about Wigner's final measurement:
If one requires the friend's measurement to yield definite results and assumes this to lead to a collapse, then one obtains measurement probabilities different from the ones derived from the unitary evolution of an isolated system~$F\otimes S$.

The incommensurability of definite measurement results with the unitary evolution of isolated system is not so much a peculiarity---or defect---of quantum mechanics.
Instead, we argue that it appears in theories that 
\begin{enumerate*}[label={(\alph*)}]
  \item\label{r1} account for interactions so that they are empirically significant, 
  \item\label{r2} require that an observation necessarily goes with \emph{such} an interaction,   
  \item\label{r4} are falsifiable, and
  \item\label{r3} in which experimental results have a minimal stability.
\end{enumerate*}
The first two requirements render an observation itself empirically traceable.
They are combined in the \emph{interaction assumption}: 
\begin{enumerate}[label={(IntA)}]
  \item\label{IntAssum} Interactions are empirically traceable. An observation necessitates such an interaction.
\end{enumerate}
The last requirement~\ref{r3} is a generalization of Popper's characterization of physics being concerned with reproducible effects\footnote{\blockcquote[{\S}I.8, emphasis in original]{Popper1958}[.]{Indeed the scientifically significant \emph{physical effect} may be defined as that which can be regularly reproduced by anyone who carries out the appropriate experiment in the way prescribed}

 } to the demand that asking the same question twice will yield the same answer:
\begin{enumerate}[label={(ISys)}]
  \item\label{IsolSys} There exist conditions under which two equivalent, subsequent measurements performed on the same system yield the same answer.
    These conditions are independent of the questions asked.
\end{enumerate}
A system satisfying these conditions will be called \emph{isolated}.\footnote{The notion, eventually, matches the common understanding of an isolated system.
We do not mean to answer the question whether \emph{perfectly isolated systems can ever be observed}. 
Indeed, we are rather doubtful that macroscopic perfectly isolated systems will ever occur in experimental setups.
Conceptually, the notion of an isolated system is important for quantum mechanics, and for the measurement problem.
These consideration place this discussion in a conceptual realm.}

In Section~\ref{sec:describing_systems}, we introduce the formal means to describe measurements in terms of lattices.
In Section~\ref{sec:interacting_systems}, we examine how to ensure that the interaction assumption~\ref{IntAssum} is satisfied.
This leads to a dichotomy between interacting systems, as described in Section~\ref{sec:interacting_systems}, and isolated systems, as described in Section~\ref{sec:isolated_systems}.
In Section~\ref{sec:interactions_within_a_joint_system}, we combine these notions and investigate how to account for an interaction \emph{within} an isolated systems, i.e., an interaction between two sub-systems of an isolated (joint) system. 
The measurement problem then arises when one attempts to impose the interaction assumption while upholding the idea of a \emph{spectator theory}, as discussed in Section~\ref{sec:spectator_theories}.

 \section{Describing Systems}\label{sec:describing_systems}
\noindent
We discuss how to abstractly represent measurements in light of falsifiability and~\ref{IsolSys}.
We follow the path of quantum logic~\cite{BirkhNeumLQM,Piron1976,bacciagaluppi2009,Coecke2000,sep-qmlogics,Moretti2019}---without actually referring to quantum mechanics---, and rely on notions of ordered sets and lattices as briefly summarized in the Appendix.
Along the way, we put the program into a new perspective.

We think of a measurement as an inquiry about a \emph{binary question}, i.e., about a question with two possible answers~$\Bt$ and~$\Bf$.\footnote{The binary questions are also referred to as \emph{propositions}. 
We think of these questions as \emph{being inquired about}, and not ``asked to the system.''
The term question allows generally for more than two answers, as considered, e.g., in~\cite{Reichenbach1944,Putnam1957,Kamlah1975}.}
The binary questions are represented by elements in a set~$\Q$.
By the requirement~\ref{IsolSys}, there exists an equivalence relation on~$\Q$, such that under appropriate conditions, two equivalent questions~$\alpha\sim\beta$ with~$\alpha,\beta\in\Q$ yield equal answers.
Let us denote by~$\QC$ the corresponding set of equivalence classes.
To ensure falsifiability~\cite{Popper1934}, we assume that for any equivalence class~$a\in\QC$, there exists a unique \emph{complementary class}~$\neg a\in\QC$ such that an inquiry about any question in~$a$ yields~$\Bt$ if and only if an inquiry about any question in~$\neg a$ yields~$\Bf$.\footnote{Subsequently, we denote elements in~$\Q$ by Greek letters, elements in~$\QC$ by Latin letters, and variables representing answers to elements in~$\Q$ by the corresponding capital Latin letters.}

We assume that the elements in~$\Q$ allow for a partial temporal order:
If a measurement corresponding to~$\alpha$ is performed \emph{before} another one~$\beta$, then~$\alpha<_t\beta$.
The conditional on~$\Q$, i.e.,~$\alpha\rightarrow\beta$, is defined as follows:
\begin{quote}
  If the inquiry about~$\alpha$ yields~$\Bt$, then a subsequent inquiry about~$\beta>_t\alpha$ yields~$\Bt$.
\end{quote}
The conditional on~$\Q$ induces an order relation on the set of equivalence classes,~$a<b$ for~$a,b\in\QC$, if verifying the order relation does not break the equivalence:
If we inquire about consecutive questions
\begin{equation*}
  \alpha<_t\beta<_t\alpha'<_t\beta', \ \text{with} \ \alpha\rightarrow\beta,\ \alpha\sim\alpha', \ \beta\sim\beta',
\end{equation*}
 then we assume to still obtain equal answers for~$\alpha$ and~$\alpha'$, as well as for~$\beta$ and~$\beta'$.
Thus, we demand that the sublattice generated by~$\{a,\neg a, b, \neg b\}$ is distributive if~$a<b$.
If there exist elements~$0$ and~$1$ such that~$a\land\neg a = 0$ and~$a\lor\neg a = 1$ for all~$a\in\QC$, then the above sublattice requirement renders the complement defined above order-reversing and leaves us with an orthocomplemented lattice of classes of equivalent questions~$(Q,<)$.\footnote{If~$a<b$, then
\begin{equation*}
  b\land \neg a = (a \lor b) \land \neg a = (a \land \neg a) \lor (b\land \neg a)
\end{equation*}
by distributivity.
With~$a\land\neg a = b\land\neg b$,
\begin{equation*}
  b\land \neg a = (b \land \neg b) \lor (b\land\neg a) = b\land (\neg a \lor\neg b)\,,
\end{equation*}
thus,~$\neg a > \neg b$.}
To ensure the distributivity of the sublattice defined above, we require the lattice to be orthomodular, i.e., to satisfy the following:
\begin{equation*}
  \text{If} \ a<b, \text{then} \ a\land (\neg a \lor b) = b\,.
\end{equation*}

 \section{Interacting systems}\label{sec:interacting_systems}
\noindent
We now turn to the question how to ensure the interaction assumption~\ref{IntAssum}.
An obvious way to trace interactions is to ``ask the system whether it interacted.''
This assumes the existence of a corresponding equivalence class~$q_{\text{int}}\in\QC$.
If we inquire consecutively about~$\alpha<_t\sigma$ with~$\sigma\in q_{\text{int}}$, then we expect~$\sigma$ to yield~$\Bt$ independent of the result of the inquiry about~$\alpha$.
This should also hold for~$\neg\alpha$.
It follows that~$q_{\text{int}}> \alpha\lor \neg\alpha$, and, therefore,~$q_{\text{int}}=1$.
Thus, the interaction assumption cannot be realized by a single inquiry about questions in a special equivalence class in~$\QC$.

Yet, we can establish within~$\Q$ whether an interaction occurred by how questions \emph{relate} to one another.
In particular, if we aim to position the characteristic ``having interacted'' dichotomously to ``being isolated'' as described in~\ref{IsolSys}, then the former characteristic is expected to be relational as is the latter.
Following this path, we demand that equivalent questions~$\sigma,\sigma'\in q_{\text{int}}$ inquired about before and after an interaction corresponding to an inquiry about a question~$\alpha\in a$ with~$\sigma<_t\alpha<_t\sigma'$ do not necessarily yield the same answer independent of what the result of the inquiry about~$\alpha$ is.
For the equivalence classes this entails\footnote{Imagine inquiring about~$\sigma<_t\alpha<_t\alpha'<_t\sigma'$ where~$\sigma,\sigma'\in q_{\text{int}}, \alpha\in a,\alpha'\in \neg a$. Then,~$\sigma\not\rightarrow\sigma'$. The same is the case for~$\sigma,\sigma'\in\neg q_{\text{int}}$.}
\begin{equation}\label{eq:non-comp}
  (q_{\text{int}} \land a ) \lor (q_{\text{int}}\land \neg a) \neq q_{\text{int}}\,.
\end{equation}
That is,~$q_{\text{int}}$ is \emph{incompatible} with~$a$.
Equivalently, the sublattice generated by~$q_{\text{int}}$ and~$a$ is not distributive (see Appendix).
As compatibility in an orthomodular lattice is symmetric, the interaction corresponding to inquiries about questions in~$q_{\text{int}}$ can be traced inversely with inquiries about questions in~$a$.

To ensure that \emph{all} elements in~$\Q$ correspond to \emph{traceable} interactions, we require that in the orthomodular lattice~$\QC$, the sublattice~$\Z$ of elements \emph{compatible with all other elements} in the lattice, called the \emph{center}, contains merely~$1$ and~$0$.
The requirement for~$\QC$ to form an orthomodular lattice with trivial center is sufficient to satisfy the interaction assumption.
Subsequently, we turn to the question whether it is necessary.

\subsection{Assigning probabilities}
\label{sub:gen_born}
\noindent
The above discussion is inspired by quantum mechanics:\footnote{See, e.g.,~\cite{Moretti2019}.} The set of orthogonal projectors on a Hilbert space~$\ProjH$ forms an atomic, orthomodular lattice with the order relation
\begin{equation*}
  P<Q \ \iff \ P(\H) \ \text{is subspace of} \ Q(\H)\,.
\end{equation*}
Gleason's theorem~\cite{Gleason57} establishes a one-to-one correspondence between probability distributions over~$\ProjH$ and density matrices if~$\dim \H\geq 3$.
With projectors forming equivalence classes along a real time parameter, 
\begin{equation*}
  \Q^{\text{qm}}=\ProjH\times\mathbb{R}\,, \text{and}\ \QC^{\text{qm}}=\ProjH\, ,
\end{equation*}
quantum mechanics carries the lattice structure \emph{before} assigning probabilities.
A priori, the lattice structure of~$\QC$ is not evident: A theory does not primarily make statements about the relation~$\alpha\rightarrow\beta$.

More generally, theories yield probability distributions for time-ordered sequences of questions,\footnote{In this approach, a map, analogous to the Born rule, is an essential part of physical theories. In classical mechanics, for example, such a ``Born rule'' consists of assigning deterministic probabilities to subsets of phase space. As we discuss below, this is only possible because the subset lattice is a Boolean lattice.} i.e., 
\begin{IEEEeqnarray*}{l}
  \Prb\big( (\alpha_1, A_1), (\alpha_2, A_2), \ldots \big) \\
  \qquad \text{with} \quad \alpha_i\in\Q, \alpha_l<_t\alpha_{l+1},A_i\in\{\Bt,\Bf\}\,,
\end{IEEEeqnarray*}
where neither~$\Q$ has a pre-established conditional, nor~$\QC$ a natural lattice structure.
Then, the requirement~\ref{IsolSys} that consecutive equivalent questions
\begin{equation*}
 \alpha\sim\beta,\ \alpha<_t\beta,\ \alpha,\beta\in\Q  
\end{equation*}
yield equal answers~$A,B\in\{\Bt,\Bf\}$ translates to 
\begin{equation*}
  \Prb( A = B) = \Prb( A=1, B=1) + \Prb( A=0, B=0) = 1 \,.
\end{equation*}

Let us consider the possibility that the probability derives from a unary function 
\begin{equation*}
  \mu: \Q\to[0,1]
\end{equation*}
such that~$\mu(\alpha)$ is the probability for the answer to~$\alpha$ is~$\Bt$.
From the above formulation of~\ref{IsolSys}, it follows that 
\begin{equation*}
  \Prb(A =B) = \mu(\alpha)\mu(\beta) + (1-\mu(\alpha))(1-\mu(\beta)) = 1
\end{equation*}
which is the case if and only if~$\mu(\alpha) =\mu(\beta)=0$ or~$\mu(\alpha) =\mu(\beta)=1$.
Thus, the function~$\mu: \Q\to \{0,1\}$ takes merely two values, and is constant within an equivalence class~$a\in\QC$. 
Therefore, there is an induced function~$\mu':\QC\to\{0,1\}$.
An immediate consequence is that if the question~$q_{\text{int}}$ from above is an element~$\QC$, then inquiring about the question ``whether the system interacted'' yields either always~$\Bt$ or always~$\Bf$, independently of inquiries about any other question. 
This is a contradiction with~\ref{IntAssum}.
A theory satisfying both~\ref{IntAssum} and~\ref{IsolSys} cannot allow for an assignment of probabilities to elements in~$\Q$ independent of inquiries about other, non-equivalent questions:
The theory is contextual~\cite{KS67,sep-kochen-specker}.

To ensure a minimal detectability of inquiries and their corresponding interactions, we are lead to assume the following, similar to Heisenberg uncertainty: 
For any~$\alpha\in\Q$ there exist equivalent~$\beta_1\sim\beta_2, \beta_1<_t\alpha<_t\beta_2$ such that 
\begin{equation*}
  \Prb(B_1\neq B_2) = \sum_{A,B} \Prb( B, A, \neg B) \geq\epsilon
\end{equation*}
for some~$\epsilon > 0$.

To connect this back to the lattice formalism, we ask: Does a contextual theory satisfying~\ref{IsolSys} and~\ref{IntAssum} give rise to an orthomodular lattice on~$\QC$?
Employing the contextuality, we define the order relation~$a<b$ by
\begin{IEEEeqnarray*}{RL}
  &\cProb{(\beta,\Bt)}{(\alpha,\Bt)} = 1 \\
  &\qquad \forall \alpha \in a, \beta\in b\,,\quad \text{and} \\
  & \Prb(A_1 = A_2) = \sum_{A,B} \Prb( A, B, A) = 1 \\
  &\qquad \forall \alpha_i\in a, \beta\in b, \alpha_1<_t\beta<_t\alpha_2\,,
\end{IEEEeqnarray*}
and the complement~$\neg a$ by
\begin{IEEEeqnarray*}{RL}
  \alpha'\in\neg a \ \iff \ &\cProb{(\alpha',\Bf)}{(\alpha,\Bt)}=1 \ \text{and} \\
  &\cProb{(\alpha',\Bt)}{(\alpha,\Bf)}=1\,.
\end{IEEEeqnarray*}
Assuming falsifiability, we demand that the complement exists and is unique.
While this yields a complemented and weakly modular poset, it is not clear whether it also constitutes an orthomodular lattice.
For the remainder of this text, we assume that~$\Q$ forms an orthomodular lattice.

The converse of the above consideration is: What probability distributions can be assigned to an orthomodular lattice?
Let us assume that~$\QC$ forms such a lattice and that~$\mu':\Q\to[0,1]$ is a function that satisfies
\begin{IEEEeqnarray*}{c}
  \mu'(0) = 0\,, \qquad \mu'(1) = 1\,;\\
  \text{if} \ a,b \ \text{are compatible, then} \\ \mu'(a) + \mu'(b) = \mu'(a\land b) + \mu'(a\lor b)\,; \\
  \text{if} \ \mu'(a_i) = 1 \ \text{then} \ \mu'\big(\land_i a_i\big) = 1\,.
\end{IEEEeqnarray*}
If we impose~\ref{IsolSys}, then, with the same reasoning as above,~$\mu'$ is a \emph{dispersion-free state}~\cite{JP1963}.
From Theorem~I in~\cite{JP1963} and Theorem~1 in~\cite{Gudder1968}, it follows that if there exists a dispersion-free state on~$\QC$ then the center~$\Z$ is not trivial.
\emph{Therefore, if we require~\ref{IntAssum} and~\ref{IsolSys}, then any assignment of probabilities to elements in~$\Q$ must be contextual.}

 \section{Isolated systems}\label{sec:isolated_systems}
\noindent
After the discussion in Section~\ref{sec:interacting_systems}, we are now able to explicate the notion of an \emph{isolated system} consistent with the two assumptions \ref{IsolSys} and \ref{IntAssum}:
\emph{A system is isolated if and only if inquiries about any two equivalent questions $\alpha\sim\beta,\alpha,\beta\in\Q$ yield equal answers with certainty.}

To empirically verify whether a system is isolated---at least for the time between two inquiries---, one inquires about any two equivalent questions and compares the thus obtained answers.
If the answers differ, then the inquiries detect an intermediate inquiry about a non-compatible question, and the system is \emph{not isolated}.
While the equality of the answers is necessary, it is, however, \emph{not sufficient} for the system to be isolated.

This empirical test is an essential ingredient in a key-distribution protocol like~\cite{BB84}.
Inversely, any theory satisfying the assumption~\ref{IntAssum} and~\ref{IsolSys} allows for a similar protocol.

Note the interdependence of the equivalence relation and the notion of a system being isolated:
If a system is isolated, then we can empirically verify the equivalence relation.
For a system with an equivalence relation, we can empirically verify whether it is isolated.
Conversely, we cannot say whether a system is isolated without a pre-established equivalence relation, and, vice versa, we cannot verify the equivalence relation without the system being isolated.
 \section{Interactions within a joint system}\label{sec:interactions_within_a_joint_system}
\noindent
If we assume the lattice structure as discussed in Section~\ref{sec:describing_systems}, then we can explicate the above considerations.
In~\cite{Piron1976}, Piron shows that an orthomodular lattice has a trivial center if and only if it is irreducible, i.e., the lattice cannot be written as a direct union, defined as follows:
The direct product of orthocomplemented lattices~$L_i$ with~$i\in I$, forms another orthocomplemented lattice~$L^p$ with the order relation 
\begin{equation*}
  x > y, x,y\in L^p \ \Leftrightarrow \ x_i > y_i \ \forall i\in I
\end{equation*}
and the orthocomplementation
\begin{equation*}
  \neg x = (\neg x_1, \ldots, \neg x_i, \ldots)\,.
\end{equation*}
It follows from~\ref{IntAssum} that the lattice~$\QC_c$ cannot be the direct product of lattices~$\QC_1$ and~$\QC_2$.

We imagine~$S_2$ to be a friendly experimenter measuring~$S_1$, inspired by the Wigner's-friend experiment~\cite{wigner1963problem, deutsch1985quantum, Wigner1961}.
Let us, for now, merely consider~$S_1$:
Before and after our friend inquires about a non-trivial~$\alpha\in\Q_1$ we inquire about two equivalent~$\alpha'\sim\alpha''$ that belong to an equivalence class incompatible to the one represented by~$\alpha$. 
The joint system~$S_1\times S_2$ is however isolated.
Thus, the equivalence classes of the joint system are not induced by the subsystems if they interact: 
Despite,~$\alpha'\sim\alpha''$ in~$\Q_1$,~$(\alpha',1)\not\sim(\alpha'',1)$ in~$\Q_c$.
The interaction between the subsystems shows in the equivalence classes of the joint system. 

Let us characterize the friend's inquiry about a non-trivial~$\alpha\in a\in\QC_1$, more specifically, as follows: \begin{quote}
  If~$(\alpha_1,\sigma)<_t(\alpha_2,\beta)\in\Q_c$ with~$\alpha_i,\alpha\in a\in\QC_1$, then~$\alpha_1\leftrightarrow\alpha_2$ and~$\alpha_1\leftrightarrow\beta$, for some~$\sigma$.
\end{quote}
\emph{A measurement effects an implication that reaches across systems.}\footnote{To illustrate this characterization, we resort to the relative-state formalism~\cite{everett1957relative} of quantum mechanics. 
The unitary describing a measurement is characterized by the following effect:
If an observer performs a measurement $\Pi_{\phi}\in\Proj{\H_S}$ on $S$, the observer is initially attested to be in a ``ready-state,'' and $S$ is prepared in a state $\phi\in\H_S$, then the observer ends up in a state $\psi\in\H_O$ of ``having obtained the correct result.''
If an observer performs a measurement $\Pi_{\phi}\in\Proj{\H_S}$ on $S$, the observer is initially attested to be in a ``ready-state,'' and $S$ is prepared in a state $\phi^{\perp}\in\H_S$ \emph{orthogonal} to $\phi$, then the observer ends up in a state $\psi^{\perp}\in\H_O$ orthogonal to $\psi$.}
It is a case not accounted for in a product lattice.
In particular, the measurement establishes the equivalence between~$(\sigma,\alpha_1)$ and~$(\beta,\alpha_2)$ in~$\Q_c$ while~$\sigma$ and~$\beta$ might not be equivalent in~$\Q_2$. 
Let us denote~$m\in\QC_c$ the equivalence class of~$(\sigma,\alpha_1)$ and~$(\beta,\alpha_2)$.
The characterizations also implies:~$(a,1)$ and~$(a,0)$ are equivalence classes in~$\QC_c$ with~$(a,0)<m<(a,1)$.
In particular, the equivalence class~$n\in\QC_c$ represented by~$(\sigma,\alpha_1')$ is incompatible with~$(a,0)$ and~$(a,1)$ if~$a'\in\QC_1$ is incompatible with~$a\in\QC_1$.\footnote{See Lemma~\ref{lm:inc_elem} in the Appendix.
In quantum mechanics, this corresponds to preparing a superposition state with respect to the friend's measurement basis.} 
Therefore, also~$n$ and~$m$ are incompatible.

To empirically test whether the two subsystems~$S_1$ and~$S_2$ interact with one another, one empirically tests the equivalence relation on~$Q_c$ by inquiring about questions in the same equivalence class and verifying that their answers match.
That is, we test whether the joint system \emph{is isolated under this equivalence relation}~(see Section~\ref{sec:isolated_systems}) and, therefore, whether equivalent questions yield same answers, \emph{independent of the choice of the equivalence class}.
Imagine, we initially inquired about a question in~$n$.
To verify that~$S_c$ is isolated, and, thus, the two subsystems interacted, we inquire about a later element in~$n$.
\emph{By the incompatibility of~$n$ and~$(a,1)$, we cannot at the same time empirically test whether the system interacted, and know about the result of the measurement.}

We encounter the measurement problem: We cannot meaningfully---i.e., with the suitable empirical support---speak of the measurement as an interaction between two systems, while maintaining the idea of the measurement yielding definite results.

\iffalse
If~$\sigma\in s<b$, then~$(a,0)<m<(a,b)<(a,1)$ and, thus,~$(a,1)$,~$(a,b)$, and~$m$ are compatible by orthomodularity.
If~$\sigma\sim\beta$, then~$m=(a,b)$.
While~$(a,0)<m<(a,1)$ is the case for any choice of the initial~$\sigma$,~$m<(a,b)$ does not hold generally.
By weak modularity in~$\QC_c$, 
Let us fix~$\sigma=\beta$, and consider another measurement,~$m'$, for which~$\neg \alpha_1\rightarrow\neg\beta$.

Then,~$(\neg a,0)<m'<(\neg a,1)$.
\fi

 \section{The Epistemological Measurement Problem}\label{sec:spectator_theories}
\noindent
The measurement problem unfolds if we compromise the interaction assumption~\ref{IntAssum} in order to save the measurement and its result from contextual dependencies.
Thus, the measurement problem exposes the idea that we can \emph{read off} measurement results without effects for the measured system---the idea of a \emph{spectator theory} described by Dewey as follows:
\blockcquote[{\S}1, p.~26]{DeweyQFC}[.]{The theory of knowing is modelled after what was supposed to take place in the act of vision.
The object refracts light to the eye and is seen; it makes a difference to the eye and to the person having an optical apparatus, but none to the thing seen. The real object is the object so fixed in its regal aloofness that it is a king to any beholding mind that may gaze upon it.
A spectator theory of knowledge is the inevitable outcome}
 A prominent explication of this epistemological model stands at the beginning of the article by Einstein, Podolski, and Rosen~\cite{EPR}:\footnote{The stance does not entirely comply with Einstein's take of a \emph{system} (see the discussion in~\cite{HW19L}). In fact, Einstein's involvement in the so-called EPR article remains debated.}
\blockcquote{EPR}[.]{If, without in any way disturbing a system, we can predict with certainty (i.e., with probability equal to unity) the value of a physical quantity, then there exists an element of physical reality corresponding to this physical quantity}

The quest for certain measurement results leads to an epistemic problem:
If knowledge is scientific knowledge and science is natural science, then the anchor of our knowledge is observation and measurement. 
And if knowledge must be constituted of certainties, then these observations cannot carry contextual dependences.\footnote{The spectator theory relates to correspondence theories of truth, as Habermas points out: \blockcquote[{\S}II, p.~68f, emphasis in original]{HabermasEuI_en}[.]{\emph{The meaning of knowledge itself becomes irrational}---in the name of rigorous knowledge.
In this way the naive idea that knowledge \emph{describes} reality becomes prevalent.
This is accompanied by the copy theory of truth, according to which the reversibly univocal correlation of statements and matters of fact must be understood as isomporphism.
Until the present day this objectivism has remained the trademark of a philosophy of science that appeared on the scene with Comte's positivism}
 The criticism of spectator theories relates to critiques of correspondence theories of truth as, e.g., in~\cite{Nietzsche1873,WittgPhiloUntersuchungen,Sellars1956,RortyCIS}.
In this context, the idea of reducing \emph{doing physics} to sole \emph{description}, as required for the measurement problem~\cite{HW19L}, becomes problematic.}

Before Galileo, the realm of certainties were reserved for the ``rational sciences'' excluding explicitly empirical or observational sciences.
The latter were, in fact, taken to come short of a `science.' 
Dewey expounds how with Galileo this separation was toppled and physics laid claim on the realm of certain knowledge.
Galileo achieved the acceptance of an empirical science into the circle of bearers of universal knowledge by synthesising experimental findings to mathematical statements that qualified as platonic truths.
\blockcquote[{\S}4, p.92]{DeweyQFC}[.]{The work of Galileo was not a development, but a revolution.
It marked a change from a qualitative to the quantitative or metric; from the heterogeneous to the homogeneous; from intrinsic forms to relations; from aesthetic harmonies to mathematical formulae; from contemplative enjoyment to active manipulation and control; from rest to change; from eternal objects to temporal sequence.
[\ldots]
But---and this `but' is of fundamental importance---in spite of the revolution, the old conceptions of knowledge as related to an antecedent reality and of moral regulation as derived from properties of this reality, persisted}
The idea that an antecedent and fixed reality can simply be \emph{read off} could be upheld because the subsets of phase space ordered by the set-inclusion form a Boolean lattice.

Quantum theory, however, brought about a substantial challenge to maintaining a spectator-theory of knowledge.\footnote{Bohr expresses in~\cite{Bohr1929a} a similar view on the epistemological implications of quantum theory:
\blockcquote[p.~115, as reprinted in \cite{BohrCW1985}]{Bohr1929a}[.]{The discovery of the quantum of action shows us, in fact, not only the natural limitation of classical physics, but, by throwing a new light upon the old philosophical problem of the objective existence of phenomena independently of our observations, confronts us with a situation hitherto unknown in natural science.
As we have seen, any observation necessitates an interference with the course of the phenomena, which is of such a nature that it deprives us of the foundation underlying the causal mode of description.
The limit, which nature herself has thus imposed upon us, of the possibility of speaking about phenomena as existing objectively finds its expression, as far as we can judge, just in the formulation of quantum mechanics.
However, this should not be regarded as a hindrance to further advance; we must only be prepared for the necessity of an ever extending abstraction from our customary demands for a directly visualizable description of nature}
 }
\blockcquote[{\S}8, p.~195f]{DeweyQFC}[.]{The principle of indeterminacy thus presents itself as the final step in the dislodgement of the old spectator theory of knowledge.
It marks the acknowledgment, within scientific procedure itself, of the fact that knowing is one kind of interaction which goes on within the world}
Our findings in Section~\ref{sec:interactions_within_a_joint_system} strengthen this insight: 
We expect any theory empirically warranting our contact with the world out there to be in conflict with the idea of a fully accessible truth-out-there.
Galileo's \emph{credo} that truth be sought in nature~\cite{Niederer82} must then be weakened to, e.g., a more structuralist or subjectivist version.
In either case, quantum mechanics draws legitimacy for the claim to represent this structure from a basis of experimental findings that we \emph{agree} upon.
This entails that different observers can make statements about \emph{equivalent} measurements performed on the \emph{same} system.
If one holds an antecedent structure of the world accountable for the agreement on measurements and on the reference to systems, then our access to this structure must evade contextuality.
Our access to this structure remains empirically unwarranted.
We are, again, left with a dualism.
 \section{Conclusion}\label{sec:conclusion}
\noindent
The requirement that there are \emph{isolated systems} for which inquiries about equivalent questions yield equal answers combined with the demand for \emph{traceable interactions} leads to \emph{contextual} theories.
With measurements corresponding to an interaction between two systems, we face a measurement problem:
We cannot meaningfully---i.e., with the suitable empirical support---speak of the measurement as an interaction between two systems while maintaining the idea of the measurement yielding definite results.
As such, the measurement problem exposes an incompatibility between epistemological stances beyond quantum mechanics:
The idea that knowledge is constituted of absolute certainties entails a spectator theory.
A spectator theory, does, however, preclude empirical evidence for interactions during measurements, i.e., for evidence that our knowledge has support from \emph{outside of us}, support from our contact with the world around us.

The incompatibility turns into a tension \emph{within} positivism:
The idea of an external source of our knowledge conflicts with the adherence to a spectator theory needed to ensure absolute certainty about that knowledge. 
 \begin{acknowledgments}
\noindent
  This work is supported by the Swiss National Science Foundation (SNF), and the \emph{NCCR QSIT}.
  We would like to thank \"Amin Baumeler, Veronika Baumann, Cecilia Boschini, Paul Erker, Xavier Coiteux-Roy, Claus Beisbart, Manuel Gil, and Christian W\"{u}thrich for helpful discussions.
\end{acknowledgments}

\appendix
\section*{Appendix: A brief introduction\\ to lattices}\label{sec:proofs}
\noindent
\begin{definition}[Lattice]
  A partially ordered set~$(L,<)$ with unique greatest lower bound~$a\land b$ and unique least upper bound~$a\lor b$ is called a \emph{lattice}.
  \begin{enumerate}[label={(\alph*)}]
  \item If the lattice has a minimal element~$0$ and a maximal element~$1$, then the lattice is called \emph{bounded}.
  \item An \emph{orthocomplementation} is a map~$a\mapsto\neg a$ such that
    \begin{equation*}
      a \land\neg a = 0, \ a\lor \neg a = 1, \ \text{if}\ a<b: \neg a > \neg b\,.
    \end{equation*}
  \item In an orthocomplemented lattice~$L$, two elements~$a,b$ are \emph{orthogonal}, denoted~$a\perp b$, if~$a<\neg b$, or, equivalently,~$b<\neg a$---using the order reversing property of the ortho-complementation.
  \item An orthocomplemented lattice~$L$ is \emph{orthomodular} if it satisfies \emph{weak modularity}, i.e., for all~$a<b$: 
    \begin{equation*}
      a\lor (\neg a\land b) = a\,,
    \end{equation*}
    or, equivalently, if~$a>b$:
    \begin{equation*}
      a\land (\neg a\lor b) = a\,.
    \end{equation*}
  \item A lattice~$L$ is \emph{distributive} if for any~$a,b,c\in L$
    \begin{equation*}
      a\land (b\lor c) = (a\lor b) \land (a\lor c)\,.
    \end{equation*}
  \end{enumerate}
\end{definition}

\begin{definition}[Compatibility]
  Two elements in a lattice,~$a,b\in L$ are called \emph{compatible}, if the lattice generated by~$\{a,\neg a, b,\neg b\}$ is distributive.
\end{definition}
If for~$a,b\in L$
\begin{equation*}
  a\land (\neg a\lor b) = b \qquad \neg b\land (b \lor \neg a) = \neg a\,,
\end{equation*}
then the elements 
\begin{equation*}
  0,a,\neg a, b, \neg b, a\land\neg b, \neg a\lor b, 1
\end{equation*}
form the distributive sublattice generated by~$\{a,\neg a,b,\neg b\}$.

\begin{theorem}[Compatibility I]\label{thm:compI}
  In an orthomodular lattice~$L$,~$a$ and~$b$ are compatible if and only if 
  \begin{equation}\label{eq:comp}
    (a\land b) \lor (\neg a\land b) = b\,.
  \end{equation}
\end{theorem}
The proof consists of combining Theorem~(2.15) and~(2.17) in~\cite{Piron1976}.
\begin{proof}
  If the sublattice is distributive, then 
  \begin{equation*}
    (a \land b) \lor (a \land \neg b) = a \quad \text{and} \quad (a \land b) \lor (\neg a \land b) = b\,.
  \end{equation*}
  It remains to show that~\eqref{eq:comp} is sufficient. 
  First we show that compatibility is symmetric: If~$(a\land b)\lor(\neg a\land b)=b$, 
  then, with~$a\land(a\lor c) =a$ for any~$c$,
  \begin{IEEEeqnarray*}{RL}
    a\land \neg b &= a\land (\neg a\lor\neg b) \land (a\lor\neg b) \\
    &= a\land (\neg a\lor\neg b) = a\land \neg (a\land b)\,.
  \end{IEEEeqnarray*}
  As~$a> (a \land b)$, we employ orthomodularity,
  \begin{equation*}
    a = (a\land b) \lor (a \land \neg( a\land b) ) = (a\land b) \lor (a \land \neg b)\,,
  \end{equation*}
  which proves the symmetry of compatibility.
  Further, we have to show equalities of the form
  \begin{equation}\label{eq:distr}
    (a\land b) \lor\neg b = a \land\neg b\,.
  \end{equation}
  Note, first, that in any lattice~$a\land b < b$ and, therefore,
  \begin{equation*}
    \underbrace{(a \land b) \lor\neg b}_{=:c_1} < \underbrace{a\lor\neg b}_{=:c_2}\,.  
  \end{equation*}
  Applying orthomodularity, i.e.,~$c_2\land ( \neg c_2 \lor c_1) = c_1$, yields
  \begin{equation*}
    (a\lor \neg b) \land \big( \underbrace{(a\land b) \lor (\neg a\land b) }_{= b} \lor \neg b \big) = (a\land b) \lor\neg b
  \end{equation*}
  and, thus, we obtain~\eqref{eq:distr}.
\end{proof}
Note the important role of weak modularity in the proof above.

\begin{theorem}[Compatibility II]\label{thm:compII}
  In an orthomodular lattice~$L$,~$a$ is compatible with~$b$ if and only if there exists mutually orthogonal elements~$a',b',c\in L$ such that 
  \begin{equation*}
    a=a'\lor c \qquad b=b'\lor c\,.
  \end{equation*}
\end{theorem}
\begin{proof}
  If~$a$ and~$b$ are compatible, then
  \begin{IEEEeqnarray*}{RL}
    (a\land b)\lor(\neg a\land b) &= b \\
    (a\land b)\lor(a\land\neg b) &= a \\
  \end{IEEEeqnarray*}
  and thus we can set 
  \begin{equation*}
    c:=a\land b \quad a':=a\land\neg b \quad b':=b\land\neg a
  \end{equation*}
  which yields the orthogonal elements.
  If, inversely,~$a$ and~$b$ can be expressed in the above form, then
  \begin{IEEEeqnarray*}{RL}
    \neg a\land b &= \neg a' \land \neg c \land (b'\lor c) \\
    &= b'\land \neg a' = b'
  \end{IEEEeqnarray*}
  using weak modularity.
  Furthermore,
  \begin{IEEEeqnarray*}{RL}
    a\land b &= (a'\lor c) \land (b'\lor c) \\
    &< (a'\lor c)\land(\neg a' \lor c)\\
    &= (a'\lor c)\land\neg a' = c
  \end{IEEEeqnarray*}
  using weak modularity again,
  and, thus, with 
  \begin{equation*}
    a\land b = (a'\lor c) \land (b'\lor c) > c 
  \end{equation*}
  we obtain~$a\land b= c$. 
  This yields~\ref{eq:comp}.
\end{proof}

\begin{lemma}[Incompatible elements]\label{lm:inc_elem}
  If~$a<b$ and~$a$ is not compatible with~$c$ for some~$a,b,c\in L$, then also~$b$ and~$c$ are incompatible.
\end{lemma}
\begin{proof}
  We have to show: If~$a<b$ and~$b$ is compatible with~$c$, then also~$a$ is compatible with~$c$.
  Let us consider
  \begin{IEEEeqnarray*}{RLL}
    a\land \neg c  &= a\land \neg c' \land \neg d \qquad &\text{Thm~\ref{thm:compII}}: c=c'\lor d \\ 
    &= a\land b'\land  \neg d &c'\perp b'\\
    &= a\land b' & d\perp b'\\
    &= a\land b = a &a<b
  \end{IEEEeqnarray*}
  and thus
  \begin{IEEEeqnarray*}{RL}
    (a\land c) \lor (a\land\neg c ) &= (a\land c) \lor a = a \,.
  \end{IEEEeqnarray*}
\end{proof}
 
\end{document}